\documentclass[12pt]{article}
\usepackage{amsmath}
\usepackage{amssymb}
\usepackage{amsfonts}
\usepackage{latexsym}
\usepackage{color}
\usepackage{graphicx}

\catcode `\@=11 \@addtoreset{equation}{section}

\catcode `\@=12

\newcommand{\be}{\begin{equation}}
\newcommand{\en}{\end{equation}}
\newcommand{\bea}{\begin{eqnarray}}
\newcommand{\ena}{\end{eqnarray}}
\newcommand{\beano}{\begin{eqnarray*}}
\newcommand{\enano}{\end{eqnarray*}}
\newcommand{\bee}{\begin{enumerate}}
\newcommand{\ene}{\end{enumerate}}

\newcommand{\mult}{\,{\scriptstyle \square}\,}

\newcommand{\mb}{\mathbb}

\newcommand{\Hil}{{\cal H}}

\newcommand{\F}{{\cal F}}
\newcommand{\B}{{\cal B}}
\newcommand{\Lc}{{\cal L}}
\newcommand{\LL}{{\cal L}}

\newcommand{\D}{{\cal D}}

\newcommand{\M}{{\cal M}}
\newcommand{\A}{{\cal A}}
\newcommand{\Ao}{{\cal A}_0}
\newcommand{\up}{\upharpoonright}

\newcommand{\bfn}{\mathbf{n}}
\newtheorem{thm}{Theorem}

\newtheorem{lemma}[thm]{Lemma}
\newtheorem{prop}[thm]{Proposition}
\newtheorem{defn}[thm]{Definition}

\newenvironment{proof}{\noindent {\bf Proof:}}{\hfill$\Box$}

\newcommand{\ip}[2]{\langle {#1}|{#2}\rangle}

\begin{document}

\thispagestyle{empty}

\vspace*{1cm}

\begin{center}
{\Large \bf Representations and derivations of quasi *-algebras induced by {\em local modifications} of states}   \vspace{2cm}\\

{\large F. Bagarello}\\
  Dipartimento di Metodi e Modelli Matematici,
Fac. Ingegneria, Universit\`a di Palermo, I-90128  Palermo, Italy\\
e-mail: bagarell@unipa.it
\vspace{5mm}\\
{\large A. Inoue}\\
Department of Applied Mathematics, Fukuoka University, Fukuoka
814-0180, Japan\\ email: a-inoue@fukuoka-u.ac.jp\\
\vspace{5mm}{\large C. Trapani}\\ Dipartimento di Matematica ed
Applicazioni, Universit\`a di Palermo,\\ I-90123 Palermo
(Italy)\\e-mail:
trapani@unipa.it\\
\end{center}

\vspace*{2cm}

\begin{abstract}
\noindent
The relationship between the GNS representations associated to states on a quasi *-algebra, which are {\em local modifications} of each other (in a sense which we will discuss) is examined. The role of local modifications on the spatiality of the corresponding induced derivations describing the dynamics of a given quantum system with infinite degrees of freedom is discussed.
\end{abstract}

\vspace{2cm}


\vfill

\newpage

\section{Introduction and preliminaries}

In two recent papers, \cite{bit1,bit2}, we have investigated the
role of derivations of quasi *-algebras and the possibility of
finding a certain symmetric operator
which {\em implements the derivation}, in the sense that in a suitable
representation the derivation can be written as a
commutator with an operator which in the physical literature is usually called the {\em effective hamiltonian}. This is useful for physical applications and produces an
algebraic framework in which the time evolution of some physical
model can be analyzed, \cite{fbrev}.

Here we continue our analysis, taking inspiration again from
physical motivations: it is  known \cite{sew} that in a
physical context {\em local modifications do not affect much the
main physical results}. Our interest here is to understand this
statement more in detail, mainly in the framework of quasi
*-algebras, \cite{schmu,ait_book}, which, as we have discussed in several other places, see
\cite{ait_book, fbrev,ctrev}, in our opinion play an important role in the
mathematical description of quantum mechanical systems with infinite
degrees of freedom.

\vspace{2mm}

Just as an introductory example, let us consider a C*-algebra $\A$ with unit $e$, and let $\omega$ and $\omega'$
be two (different) positive linear functionals on $\A$. Let further $(\pi_\omega,\xi_\omega,\Hil_\omega)$ and
$(\pi_{\omega'},\xi_{\omega'},\Hil_{\omega'})$ be their associated GNS-representations. An interesting problem is the following: under which conditions on $\omega$ and $\omega'$ are the representations $\pi_{\omega}$ and $\pi_{\omega'}$ unitarily equivalent?

It is somehow more convenient to consider first the following preliminary problem: {\em how  must $\omega$ and $\omega'$ be related  for $\pi_{\omega'}$ to be unitarily equivalent to a
sub *-representation of $\pi_{\omega}$} ? An easy proof shows that

 {\em $\pi_{\omega'}$ is unitarily equivalent to a sub *-representation of $\pi_\omega$ if, and
only if, there exists a sequence $\{b_n\}$ of elements of $\A$ such that
$\omega'(a)=\lim_{n\to\infty}\omega(b_n^* a b_n)$  $\forall\,a\in\A$, and the sequence $\{\pi_\omega(b_n)\xi_\omega\}$ converges in
$\Hil_\omega$.}

We refer to \cite{sew} for the physical implications of this result. Here we observe that, in particular, if $\omega$ is a positive linear functional on $\A$,
and $b\in\A$ a fixed element such that $\omega(b^*b)\neq 0$,
then the GNS-representation associated to
$\omega_b(\cdot)=\omega(b^*\cdot b)$ is unitarily
equivalent to a sub *-representation of $\pi_\omega$. This means
that there exists a subspace $\Hil_\omega^b$ of $\Hil_\omega$ and a
unitary operator $U: \Hil_{\omega_b}\to\Hil_\omega^b$ such that
$\pi_{\omega_b}(a)=U^{*}\pi_\omega^b(a)U$ for all $a\in\A$, where $\pi_\omega^b(a):=\pi_{\omega}(a)\upharpoonright
 _{\Hil_\omega^b}$.

 Going back to our original question, i.e. to the unitary equivalence of $\pi_\omega$ and $\pi_{\omega'}$, we will postpone this analysis to the next section, where the more relevant case of quasi *-algebras is discussed.

\vspace{2mm}

 Let now $\delta$ be a *-derivation on $\A$ and let us define $\delta_{\pi_{\omega}^b}
 (a)=\pi_{\omega}^b(\delta(a))$ and $\delta_{\pi_{\omega_b}}
 (a)=\pi_{\omega_b}(\delta(a))$, $a\in\A$.  The first obvious remark
 is that, under our assumptions,
 $$
\delta_{\pi_{\omega_b}}(a)=\pi_{\omega_b}(\delta(a))=U^*\pi_{\omega}^b(\delta(a))U=U^*\delta_{\pi_{\omega}^b}
 (a)U.
 $$
Secondly, if  $\delta_{\pi_{\omega}^b}(a)$ is spatial,
i.e.  there exists an element $H_{\pi_{\omega}^b}\in
B(\Hil_{\omega}^b)$ such that
$\delta_{\pi_{\omega}^b}(a)=i[H_{\pi_{\omega}^b},\pi_{\omega}^b(a)]$,
$a\in\A$, then  $\delta_{\pi_{\omega_b}}$ is also spatial and the
implementing operator is
$H_{\pi_{\omega_b}}=U^*H_{\pi_{\omega}^b}U$, which belongs to
$B(\Hil_{\omega_b})$.

>From a physical point of view we can interpret this result as
follows: it is well known that no hamiltonian operator exists in
general which implements the dynamics of an infinitely extended system,
\cite{sew}. For this reason one has to consider  a finite-volume approximation of
the system, for which a self-adjoint energy operator $H_V$ can be
defined. Associated to $H_V$ we can introduce a finite-volume
derivation $\delta_V(X)=i[H_V,X]$, for each observable $X$ localized
in $V$, and a time evolution $\alpha_V^t(X)=e^{iH_Vt}Xe^{-iH_Vt}$.
However, usually, neither $\delta_V(X)$ nor $\alpha_V^t(X)$ converge
in the uniform, strong or weak topology. One usually has to consider
some representation of the abstract algebra and, as in \cite{bit1},
the corresponding family of {\em effective derivations}, i.e. derivations
in the given representation. This net of derivations may now be converging
and, under suitable conditions, it still defines a derivation whose
implementing operator is  the { effective hamiltonian}. Therefore the choice of the representations in this procedure is crucial.
Our results  show that, in fact, there is no essential difference
between the effective hamiltonians that we obtain starting from two
different representations, at least if they are GNS generated by a
fixed positive linear functional $\omega$ and by a different positive linear functional $\omega'=\omega_b$,
for each possible choice of $b\in\A$. In particular this implies that, if $b$ is a
local observable (meaning by this that it belongs to some of the
$\A_V$'s which produce the quasi local C*-algebra,
\cite{sew,fbrev}), then the two related derivations are unitarily
equivalent and, consequently, the two effective hamiltonians are unitarily
equivalent as
well. Hence their physical content is essentially the same, as claimed before.

\section{The case of quasi *-algebras}

We begin this section with recalling briefly the definitions of quasi *-algebras and their *-representations and sub *-representations. More details can be found in \cite{schmu,ait_book}.

{Let $\A$ be a complex vector space and $\Ao$ a  $^\ast$ -algebra contained
in $\A$. We say that $\A$ is a quasi  $^\ast$ -algebra with
distinguished  $^\ast$ -algebra $\Ao$ (or, simply, over $\Ao$) if
\begin{itemize}\item[(i)] the left multiplication $ax$ and the right multiplication $
xa$ of an element $a$ of $\A$ and an element $x$ of $\A_0$ which
 extend the multiplication of $\A_0$ are always defined and
bilinear; \item[(ii)] $x_1 (x_2 a)= (x_1x_2 )a$ and $x_1(a
 x_2)= (x_1 a) x_2$, for each $x_1, x_2 \in \A_0$ and $a \in \A$;

\item[(iii)] an involution $*$ which extends the involution of $\A_0$
is defined in $\A$ with the property $(ax)^*= x^*a^*$ and $(xa)^
* =a^* x^*$ for each $x \in \A_0$ and $a \in \A$.
\end{itemize}
}

Let now $\D$ be a dense subspace of a Hilbert space $\Hil$. We denote by $ \Lc^\dagger(\D,\Hil) $ the set of all
(closable) linear operators $X$ such that $ {\D}(X) = {\D},\; {\D}(X^*) \supseteq {\D}.$

The set $ \LL^\dagger(\D,\Hil ) $ is a  partial *-algebra
 with respect to the following operations: the usual sum $X_1 + X_2 $,
the scalar multiplication $\lambda X$, the involution $ X \mapsto X^\dagger = X^* \up {\D}$ and the \emph{
(weak)} partial multiplication $X_1 \mult X_2 = {{X_1}^\dagger}^* X_2$, defined whenever $X_2$ is a weak right
multiplier of $X_1$ (we shall write $X_2 \in R^{\rm w}(X_1)$ or $X_1 \in L^{\rm w}(X_2)$), that is, iff $ X_2
{\D} \subset {\D}({{X_1}^\dagger}^*)$ and  $ X_1^* {\D} \subset {\D}(X_2^*).$

Let $\Lc^\dagger(\D)$ be the subspace of $\Lc^\dagger(\D,\Hil)$ consisting of all its elements  which leave, together with their adjoints, the domain $\D$ invariant. Then $\Lc^\dagger(\D)$ is a *-algebra with
respect to the usual operations.

Let $(\A,\Ao)$ be a quasi *-algebra with identity $e$ and $\D_\pi$  a dense domain in a certain Hilbert
space $\Hil_\pi$. A linear map $\pi$ from $\A$ into $\LL^\dagger(\D_\pi, \Hil_\pi)$ such that:

(i) $\pi(a^*)=\pi(a)^\dagger, \quad \forall a\in \A$,

(ii) if $a\in \A$, $x\in \Ao$, then $\pi(a)${$\Box$}\!\! $\pi(x)$ is well defined and
$\pi(ax)=\pi(a)${$\Box$}\!\! $\pi(x)$,

\noindent
is called  a *-representation of $\A$. Moreover, if

(iii) $\pi(\Ao)\subset \LL^\dagger(\D_\pi)$,

\noindent then $\pi$ is said to be a *-representation of the quasi *-algebra $(\A,\Ao)$.

\medskip If $\pi$ is a *-representation of $(\A, \Ao)$, then the {\em closure} $\widetilde{\pi}$ of $\pi$ is defined, for each
$x \in \A$, as the restriction of $\overline{\pi(x)}$ to the domain $\widetilde{\D_\pi}$, which is the completion
of $\D_\pi$ under the {\em graph topology} $t_\pi$ \cite{schmu} defined by the seminorms $\xi \in \D_\pi \to
\|\pi(a)\xi\|$, $a\in \A$. If $\pi=\widetilde{\pi}$ the representation is said to be {\em closed}.

The adjoint of a *-representation $\pi$ of a quasi *-algebra $(\A, \Ao)$ is defined as follows:

$$
\D_{\pi^*}  \equiv \bigcap_{x \in \A} \D(\pi(x)^* ) \text{ and } \pi^* (x) = \pi(x^* )^*\, \upharpoonright
\D_{\pi^*}, \quad x \in \A .$$

The representation $\pi$ is said to be {\em self-adjoint} if $\pi=\pi^*$.

The representation $\pi$ is said to be {\em ultra-cyclic} if there exists $\xi_0\in\D_\pi$ such that
$\D_\pi=\pi(\Ao)\xi_0$, while is said to be {\em cyclic} if there exists $\xi_0\in\D_\pi$ such that
$\pi(\Ao)\xi_0$ is dense in $\D_\pi$ w.r.t. $t_\pi$.

\begin{defn}

Let $\pi$ be a *-representation of $\A$. A subspace $\M\subset \D_\pi$ is said to be {{\em quasi-invariant}} for $\pi$ if
$\pi(\Ao)\M\subset\M$ and $\pi(\A)\M\subset\overline{\M}$, the closure of $\M$ in the Hilbert norm of $\Hil_\pi$.
Moreover the quasi-invariant subspace $\M$ is called {\em ultra-cyclic} if   there exists $\xi_0\in\M$ such
that $\M=\pi(\Ao)\xi_0$. $\M$ is called {\em cyclic} if  there exists $\xi_0\in\M$ such that
$\pi(\Ao)\xi_0$ is dense in $\M$ w.r.t. $t_\pi$.

\end{defn}

\begin{prop}
Let $\pi$ be a *-representation of $\A$ and $\M$ a quasi-invariant subspace of $\D_\pi$ for $\pi$. We put
$$
\left\{
\begin{array}{l}
\D_{\pi\up\M} := \M, \\
(\pi\up\M)(x):=\pi(x)\up\M,\quad x\in\A. \\
\end{array}
\right.
$$
Then $\pi\up\M$ is a *-representation of $\A$ with domain $\M$ in $\overline{\M}$. Let $\pi_\M$ denote the
closure of $\pi\up\M$. Then

(i) if $\M$ is ultra-cyclic then $\pi\up\M$ is ultra-cyclic and $\pi_\M$ is cyclic;

(ii) if $\M$ is cyclic then $\pi\up\M$  and $\pi_\M$ are cyclic.

\end{prop}

In the sequel we will also need the following definitions:

\begin{defn}
Let $\rho$ and $\pi$ be *-representations of $\A$ respectively on $\D_\rho\subset\Hil_\rho$ and $\D_\pi\subset\Hil_\pi$. Then $\rho$ and $\pi$ are unitarily equivalent if there exists
a unitary operator $U:\Hil_\rho\rightarrow\Hil_\pi$ such that $U\D_\rho=\D_\pi$ and $\rho(x)=U^*\pi(x)U$, for all
$x\in\A$.

\end{defn}

\begin{defn}
Let $\pi$ be a *-representation of $\A$. Then $\pi'$ is a sub *-representation of $\pi$ if and only if
$\pi'=\pi\up\M$, for a certain quasi-invariant subspace $\M$ of $\D_\pi$. Furthermore $\pi'$ is a closed sub
*-representation of $\pi$ if and only if  $\pi'=\pi_\M$, for a certain quasi-invariant subspace $\M$ of $\D_\pi$.

\end{defn}

The following proposition, proved by one of us in \cite{ct_ban}, extends the GNS construction to quasi *-algebras.

\begin{prop}

Let $\omega$ be a linear functional on $\A$ satisfying the following requirements:

(L1) $\omega(a^*a)\geq 0$ for all $a\in\Ao$;

(L2) $\omega(b^*x^*a)=\overline{\omega(a^*xb)}$,
$\forall\,a,b\in\Ao$, $x\in\A$;

(L3) $\forall x\in\A$ there exists $\gamma_x>0$ such that
$|\omega(x^*a)|\leq \gamma_x\,\omega(a^*a)^{1/2}$.

Then there exists a triple $(\pi_{\omega}, \lambda_{\omega}, \Hil_{\omega})$ such that

  $\bullet$ $\pi_{\omega}$ is a ultra-cyclic *-representation of $\A$ with ultra-cyclic vector $\xi_\omega$;

  $\bullet$ $\lambda_{\omega}$ is a linear map of $\A$ into
  $\Hil_{\omega}$ with $\lambda_{\omega}(\Ao)=\D_{\pi_\omega}$, $\xi_\omega=\lambda_{\omega}(e)$ and
  $\pi_{\omega}(x)\lambda_{\omega}(a)=\lambda_{\omega}(xa)$, for every $x \in\A,\, a \in \Ao$;

  $\bullet$ $\omega(x)=\ip{\pi_{\omega}(x)\xi_\omega}{\xi_\omega}$,
  for every $x \in \A$.
\end{prop}

The representation $\pi_\omega$ satisfies the  properties: (1) $\pi_{\omega_0}=\pi_\omega\upharpoonright_{\Ao}$;
(2) $\pi_\omega(x)\lambda_\omega(a)=\lambda_\omega(xa)$, $x\in \A, \, a\in \Ao$ and
(3) $\pi_\omega^*(a)\lambda_\omega(x)=\lambda_\omega(ax)$, $x\in \A, \, a\in
  \Ao$. Here $\pi_\omega^*$ denotes the adjoint representation of $\pi$, see \cite{schmu, ait_book}.

For shortness, a linear functional $\omega$ on $\A$ satisfying (L1)-(L3) will be called a {\em representable}
functional on $\A$. If $\omega$ is representable, $(\pi_\omega, \lambda_\omega, \Hil_\omega)$ will be called, as
usual, the GNS construction for $\omega$.

\vspace{1mm}

It is possible to check that conditions (L1)-(L3) are stable under the map $\omega\rightarrow \omega_b$, with $b
\in\Ao$. This means that, if $\omega$ is representable, then $\omega_b$ is representable, for every $b \in \Ao$.
We only prove (L3) since (L1) and (L2) are trivial. We have
$$
\left|\omega_b(x^*a)\right|=\left|\omega((xb)^*)ab\right|\leq
\gamma_{xb}\omega((ab)^*ab)^{1/2}=\gamma_{xb}\omega_b(a^*a)^{1/2}.
$$
Hence $\omega_b$ produces a GNS representation as well, so
that it is worth comparing the two representations arising from
$\omega$ and $\omega_b$, in view of extending to  quasi *-algebras what we discussed in the first section for C*-algebras.

We start with considering
the following question: {\em when a representable linear functional $\omega'$ can be written as
$\omega'=\omega_b$, for some $b\in\Ao$?} To answer this question we  give the following

\begin{prop}
Let $\omega'$ and $\omega$ be  representable linear functionals on $\A$. Then $\omega'=\omega_b$ for some
$b\in\Ao$ if and only if $\pi_{\omega'}$ is unitarily equivalent to a sub *-representation of $\pi_\omega$.

\end{prop}

\begin{proof}
Suppose first that $\omega'=\omega_b$ for some $b\in\Ao$. For every $x \in \A$ and $a,c \in \Ao$, we have
\begin{equation}\label{U1}
\omega_b(c^*xa)=\ip{\pi_{\omega_b}(x)\lambda_{\omega_b}(a)}{\lambda_{\omega_b}(c)}.\end{equation} On the other
hand,
\begin{equation}\label{U2}\omega_b(c^*xa)= \omega(b^*c^*xab)=
\ip{\pi_\omega(x)\pi_\omega(a)\lambda_{\omega}(b)}{\pi_\omega(c)\lambda_{\omega}(b)}.\end{equation}

Now put $\Hil_\omega^b := \overline{\pi_\omega(\Ao)\lambda_{\omega}(b)}$. Then, from  equality
\eqref{U1},  it follows that there exists a unitary operator $U:\Hil_\omega^b \to\Hil_{\omega_b}$ such that
$$U\pi_\omega(a)\lambda_{\omega}(b)= \lambda_{\omega_b}(a), \quad \forall a\in \Ao.$$

\noindent
>From \eqref{U2} we deduce that, for every $a \in \A$ and $a,c \in\Ao$,
\begin{eqnarray*}\ip{\pi_\omega(x)\pi_\omega(a)\lambda_{\omega}(b)}{\pi_\omega(c)\lambda_{\omega}(b)} &=&
\ip{\pi_{\omega_b}(x)\lambda_{\omega_b}(a)}{\lambda_{\omega_b}(c)} \\
&=&
\ip{\pi_{\omega_b}(x)U\pi_\omega (a)\lambda_{\omega}(b)}{U\pi_\omega (c)\lambda_{\omega}(b)} \\
&=& \ip{U^*\pi_{\omega_b}(x)U\pi_\omega(a)\lambda_{\omega}(b)}{\pi_\omega(c)\lambda_{\omega}(b)}.
\end{eqnarray*}
This implies that
$$\pi_\omega^b(x):= \pi_\omega(x)_{\upharpoonright \pi_\omega(\Ao)\lambda_{\omega}(b) }= U^*\pi_{\omega_b}(x)U
\upharpoonright \pi_\omega(\Ao)\lambda_{\omega}(b).$$ Hence, ${\pi_\omega(\Ao)\lambda_{\omega}(b)}$ is a quasi-invariant
subspace for $\pi_\omega$, that is, $\pi_\omega(\A){\pi_\omega(\Ao)\lambda_{\omega}(b)}\subseteq
\overline{\pi_\omega(\Ao)\lambda_{\omega}(b)}$ and so $\pi_\omega^b$ is a sub *-representation of $\pi_\omega$
with ultra-cyclic vector $\lambda_{\omega}(b)$, and it is unitarily equivalent to $\pi_{\omega_b}$.

Conversely, suppose that $\pi_{\omega'}$ is unitarily equivalent to a sub *-representation of $\pi_\omega$. Then there
exists a quasi-invariant subspace $\M$ of $\D_{\pi_\omega}$, and a unitary operator $U: \Hil_{\omega'}\rightarrow
\overline{\M}\subset \Hil_\omega$  such that
$U\lambda_{\omega'}(\Ao)=\M\subset\lambda_{\omega}(\Ao)=\D_{\pi_\omega}$ and
$\pi_{\omega'}(x)=U^*(\pi\up\M)(x)U$, $\forall x\in\A$. Since
$U\lambda_{\omega'}(e)\in\M\subset\lambda_{\omega}(\Ao)$, then there exists $b\in\Ao$ such that
$U\lambda_{\omega'}(e)=\lambda_{\omega}(b)$. Thus, for every $x\in\A$,
\begin{eqnarray*}\omega'(x)=\ip{\pi_{\omega'}(x)\lambda_{\omega'}(e)}{\lambda_{\omega'}(e)} &=&
\ip{\pi_{\omega}(x)U\lambda_{\omega'}(e)}{U\lambda_{\omega'}(e)} \\&=&
\ip{\pi_{\omega}(x)\lambda_{\omega}(b)}{\lambda_{\omega}(b)} = \omega_b(x).
\end{eqnarray*}

\end{proof}

We now consider a slightly generalized problem, looking for conditions under which a representable linear functional $\omega'$ on $\A$ can be written as
$\omega'=\lim_\alpha\omega_{b_\alpha}$ for some net $\{b_\alpha\}$ in $\Ao$.

\begin{prop}
Let $\omega'$ and $\omega$ be  representable linear functionals on $\A$. Then $\omega'=\lim_{\alpha}\omega_{b_\alpha}$ for
some net $\{b_\alpha\}$ in $\Ao$ such that $\{\pi_\omega(b_\alpha)\xi_\omega\}$ converges w.r. to
$t_{\pi_\omega}$ if, and only if,  $\pi_{\omega'}$ is unitarily equivalent to a sub *-representation of
$\tilde\pi_\omega$.

\end{prop}

\begin{proof} Suppose that $\omega'=\lim_\alpha\omega_{b_\alpha}$, for some net $\{b_\alpha\}$ in $\Ao$ such that
$\{\pi_\omega(b_\alpha)\xi_\omega\}$ converges w.r. to $t_{\pi_\omega}$. Then, it is easily shown that $\M :=
\tilde{\pi}_\omega(\Ao)\xi_0$ is a quasi-invariant subspace of $\D_{\tilde{\pi}_\omega}$, where $\xi_0:=
t_{\pi_\omega}-\lim_\alpha \pi_\omega(b_\alpha)\xi_\omega$. For every $x \in \A$ and every $a,c \in \Ao$, we have
\begin{eqnarray}
\ip{\pi_{\omega'}(x)\lambda_{\omega'}(a)}{\lambda_{\omega'}(c)}&=& \omega' (c^*xa)\nonumber\\
&=& \lim_\alpha \omega(b_\alpha^*c^*xab_\alpha)\nonumber\\
&=& \lim_\alpha\ip{\pi_\omega(xa)\lambda_\omega(b_\alpha)}{\pi_\omega(c)\lambda_\omega(b_\alpha)}\nonumber\\
&=& \ip{\tilde{\pi}_\omega(xa)\xi_0}{\tilde{\pi}_\omega(c)\xi_0}\nonumber\\
&=& \ip{(\tilde{\pi}_\omega\up\M)(x)\tilde{\pi}_\omega(a)\xi_0}{\tilde{\pi}_\omega(c)\xi_0}. \label{ex1}
\end{eqnarray}
Here we put
$$
U\tilde\pi_\omega(a)\xi_0=\lambda_{\omega'}(a), \quad a\in\Ao.
$$
Then $U$ extends to a unitary operator of $\overline{\M}$ onto $\Hil_{\omega'}$, which we denote with the
same symbol, such that $U\M=\lambda_{\omega'}(\Ao)=\D_{\pi_{\omega'}}$. Furthermore, by (\ref{ex1}), we have

\begin{eqnarray*}\ip{\pi_{\omega'}(x)\lambda_{\omega'}(a)}{\lambda_{\omega'}(c)} &=&
\ip{(\tilde\pi_{\omega}\up\M)(x)\tilde\pi_{\omega}(a)\xi_0}{\tilde\pi_{\omega}(c)\xi_0} \\&=&
\ip{(\tilde\pi_{\omega}\up\M)(x)U^*\lambda_{\omega'}(a)}{U^*\lambda_{\omega'}(c)}\\ &=&
\ip{U(\tilde\pi_{\omega}\up\M)(x)U^*\lambda_{\omega'}(a)}{\lambda_{\omega'}(c)},
\end{eqnarray*}
for each $a,c\in\Ao$ and $x\in\A$, which implies that
$$
\pi_{\omega'}(x)=U(\tilde\pi_{\omega}\up\M)(x)U^*, \qquad \forall x\in\A.
$$
Thus $\pi_{\omega'}$ is unitarily equivalent to a sub *-representation $\tilde\pi_\omega\up\M$ of
$\tilde\pi_\omega$. Conversely, suppose $\pi_{\omega'}$ is unitarily equivalent to a sub *-representation of $\tilde\pi_\omega$. Then, there exists a quasi-invariant subspace of $\D_{\tilde\pi_\omega}$, $\M$, and a unitary
operator $U: \Hil_{\omega'}\rightarrow \overline{\M}$ such that
$U\lambda_{\omega'}(\Ao)=\M\subset\D_{\tilde\pi_\omega}$ and $\pi_{\omega'(x)}=U^*(\pi_\omega\up\M)(x)U$,
$\forall x\in\A$. Since $U\lambda_{\omega'}(e)\in\M\subset\D_{\tilde\pi_\omega}$, there exists
$\{b_\alpha\}\subset\Ao$ such that $\lambda_{\omega}(b_\alpha)=\pi_\omega(b_\alpha)\xi_\omega\rightarrow
U\lambda_{\omega'}(e)$, in the  topology $t_{\pi_\omega}$. Hence,
\begin{eqnarray*}
\omega' (x) &=&  \ip{\pi_{\omega'}(x)\lambda_{\omega'}(e)}{\lambda_{\omega'}(e)}\\
&=& \ip{\pi_\omega(x)U\lambda_{\omega'}(e)}{U\lambda_{\omega'}(e)}\\
&=& \lim_\alpha \ip{\pi_\omega(x)\lambda_{\omega}(b_\alpha)}{\lambda_{\omega}(b_\alpha)}  \\
&=& \lim_\alpha \omega_{b_\alpha}(x),
\end{eqnarray*}
for every $x\in\A$.

\end{proof}

The previous propositions, and in particular Proposition 6, show that, for every $b\in\Ao$ such that
$\omega(b^*b)\neq0$,
 $\omega$ and
$\omega_b$ produce {\em close} GNS representations and the same
physical considerations given in Section I can also be repeated here, with
no major change. In particular we consider now some consequences of our results on the theory of spatial derivations in  the quasi *-algebraic setting discussed in \cite{bit1,bit2}.
To keep the paper self-contained, let us first recall few definitions. Let $(\A,\Ao)$ be a quasi *-algebra. A {\em
*-derivation of}  $\Ao$ is a map $\delta: \Ao\rightarrow \A$
 with the following properties:
\begin{itemize}
\item[(i)]  $\delta(a^*)=\delta(a)^*, \; \forall a \in \Ao$;
\item[(ii)] $\delta(\alpha a+\beta b) = \alpha \delta( a)+\beta\delta( b), \; \forall a,b
 \in \Ao, \forall \alpha,\beta \in \mathbb{C}$;
\item [(iii)] $\delta(ab) = a\delta( b)+\delta( a)b,  \; \forall a,b \in \Ao$.
\end{itemize}
Further, let $\pi$ be a
*-representation of  $(\A,\Ao)$. As in \cite{bit1} we will always assume that
whenever $a\in \Ao$ is such that $\pi(a)=0$, then $\pi(\delta(a))=0$ as
well. Under this assumption, the linear map $
\delta_\pi(\pi(a))=\pi(\delta(a)), \quad a\in \Ao$, is well-defined
on $\pi(\Ao)$ with values in $\pi(\A)$ and it is a
*-derivation of $\pi(\Ao)$. We call $\delta_\pi$ the *-derivation
{\em induced} by $\pi$. Given such a representation $\pi$  and its
dense domain $\D_\pi$, we consider the usual graph topology
$t_\dagger$ generated by the seminorms $ \xi\in\D_\pi \rightarrow
\|A\xi\|, \quad A\in \Lc^\dagger(\D_\pi)$.

If $\D_\pi'$ denotes the conjugate dual of $\D_\pi$, we get the usual rigged Hilbert space $\D_\pi[t_\dagger]
\subset \Hil_\pi  \subset \D_\pi'[t_\dagger']$, where $t_\dagger'$ is the strong dual topology of $\D_\pi'$. As
usual, we denote by $\Lc(\D_\pi,\D_\pi')$ the space of all continuous linear maps from $\D_\pi[t_\dagger]$ into
$\D_\pi'[t_\dagger']$. In this case, $\Lc^\dagger(\D_\pi)\subset \Lc(\D_\pi,\D_\pi')$. Each operator $A\in
\Lc^\dagger(\D_\pi)$ can be extended to the whole $\D_\pi'$ by putting
$$
<\hat A\xi',\eta>=<\xi',A^\dagger \eta>, \quad \forall \xi'\in
\D_\pi', \quad \eta\in \D_\pi,
$$
where $<\cdot, \cdot>$ denotes the form which puts $\D_\pi$ and $\D_\pi'$ in conjugate duality. Hence the
multiplication of $X\in\Lc(\D_\pi,\D_\pi')$ and $A\in\Lc^\dagger(\D_\pi)$ can always be defined. Indeed,
\cite{bit1}, $(X\circ A)\xi=X(A\xi), \mbox{ and } (A\circ X)\xi=\hat A(X\xi)$, $\forall \xi\in \D_\pi$.
{ With these definitions, however, $(\Lc(\D_\pi,\D_\pi'),\Lc^\dagger(\D_\pi))$ may fail to be a quasi
*-algebra, since the operator $X\circ A$ need not be continuous from $\D_\pi[t_\dagger]$ into
$\D_\pi'[t_\dagger']$, unless some additional condition, like the reflexivity of $\D_\pi[t_\dagger]$, is
fulfilled. From now on, we will assume that $\D_\pi[t_\dagger]$ is a reflexive space. This assumption (which was
missing in \cite{bit1}) even though restrictive, is fulfilled in most of the physical models considered so far,
\cite{fbrev}.

}

Given a derivation $\delta$ of $(\A,\Ao)$ and  a *-representation $\pi$ of $(\A,\Ao)$, that we suppose to be
cyclic with cyclic vector $\xi_0$, the induced derivation $\delta_\pi$ is spatial if there exists {
$H_\pi=H_\pi^\dagger\in  \Lc(\D_\pi,\D_\pi')$} such that $H_\pi\xi_0\in \Hil_\pi$ and $$
\delta_\pi(\pi(x))=i\{H_\pi\circ\pi(x)-\pi(x)\circ H_\pi\},\quad\forall x\in\Ao.$$

Let $(\A, \Ao)$ be a locally convex quasi *-algebra with locally convex topology $\tau$. In \cite{bit1} we have found necessary and sufficient
conditions  for an induced derivation to be spatial. One of these
conditions is the following:

\noindent{\em there exists a positive linear functional $f$ on $\Ao$ such that:
\begin{equation}
f(a^*a)\leq p(a)^2, \quad \forall a\in \Ao, \label{21}
\end{equation}
for some continuous seminorm $p$ of $\tau$ and, denoting with
$\tilde f$ the continuous extension of $f$ to  $\A$, the following
inequality holds:
\begin{equation}
|\tilde f(\delta(a))|\leq C(\sqrt{f(a^*a)}+\sqrt{f(aa^*)}), \quad
\forall a\in \Ao, \label{22}
\end{equation}
for some positive constant $C$.}

Suppose now that $\omega_0$ is a positive linear representable functional on $\Ao$ satisfying condition \eqref{21}. Let
$\omega:=\widetilde{\omega_0}$ be the continuous extension of $\omega_0$ to $\A$, that is
 $$\omega(x)=\lim_\alpha \omega_0(a_\alpha), \quad x \in \A,$$
 where ${a_\alpha}$ is a net in $\Ao$ which converges to $x$ w. r. to $\tau$. Then $\omega$
 automatically satisfies conditions (L1), (L2) and (L3). Indeed, (L1) is clear since
$\omega_0$ is positive by assumption. As for (L2), let $x\in\A$ and $\{x_\alpha\}\subset\Ao$ be a net $\tau$-converging
to $x$. Since $\omega_0$ is hermitian we have $\omega_0(b^*x_\alpha^*a)=\overline{\omega_0(a^*x_\alpha b)}$, for all
$a,b\in\Ao$. Because of (\ref{21}), taking the limit on $\alpha$ of this equality we get (L2). To prove (L3) we
first use the Schwarz inequality on $\Ao$: $|\omega_0(x_\alpha a)|\leq \omega_0(x_\alpha^*
x_\alpha)^{1/2}\,\omega_0(a^* a)^{1/2}$. But $\omega_0(x_\alpha^* x_\alpha)^{1/2}\leq p(x_\alpha)^2\rightarrow
p(x)^2$ so that
$$
|\omega(xa)|=\lim_\alpha|\omega_0(x_\alpha a)|\leq p(x)\,\omega(a^*
a)^{1/2}
$$
which is (L3).

\vspace{2mm}

Suppose that $\omega_0$ is a positive linear representable functional on $\Ao$
satisfying both conditions (\ref{21}) and (\ref{22}). Then we consider
the question as to whether ${(\omega_0)}_b$ satisfies these same conditions.
This is important for the following reason. If both $\omega_0$ and $(\omega_0)_b$ satisfy (\ref{21}) and (\ref{22}), then they have continuous extensions $\omega$ and $\widetilde{(\omega_0)_b}$ respectively to $\A$ and it turns out that  $\widetilde{(\omega_0)_b}=\omega_b$. Thus $\omega_b$ satisfies conditions (L1), (L2) and (L3) and both $\delta_{\pi_\omega}$ and $\delta_{\pi_{\omega_b}}$ are spatial. Hence
a relation between the effective hamiltonians can be found.

 First
we notice that, because of the continuity of the multiplication, we
have
$$
{(\omega_0)}_b(a^*a)=\omega_0((ab)^*ab)\leq p(ab)^2\leq q(a)^2, \quad a \in \Ao
$$
for some continuous seminorm $q$ of $\tau$.

{ Thus we have the following
\begin{prop}
Let $(\A,\Ao)$ be a locally convex quasi *-algebra with
locally convex topology $\tau$, $\delta$ a *-derivation of
$(\A,\Ao)$and $\omega_0$ a positive linear functional on
$\Ao$.

(1) Suppose that $\omega_0$ satisfies the condition
$$
\omega_0(a^*a)\leq p(a)^2,\qquad \forall\,a\in\Ao
$$
for some continuous seminorm $p$ of $\tau$. Then the
continuous extension $\omega:=\tilde\omega_0$ of $\omega_0$
to $\A$ and every $\omega_b$, $b\in\Ao$, produce the
ultra-cyclic GNS-representations $\pi_\omega$ and
$\pi_{\omega_b}$.

(2) Furthermore, suppose that
$$
|\omega(\delta(a))|\leq
C\left(\sqrt{\omega(a^*a)}+\sqrt{\omega(aa^*)}\right),
\quad \forall\,a\in\Ao
$$
for some positive constant $C$. Then the *-derivation $\delta_{\pi_\omega}$ induced by $\pi_\omega$ is spatial.
If \mbox{$\pi_\omega\upharpoonright\Ao$} is bounded, in particular in the case where $\Ao$ is a C*-algebra, then
the *-derivation $\delta_{\pi_{\omega_b}}$ induced by $\pi_{\omega_b}$ is also spatial for every $b\in\Ao$.
\end{prop}
\begin{proof}
We need only to prove the last statement in (2). For this we notice that if
$b\in\Ao$ is such that $\pi_\omega(b)$ is bounded, then $\omega_b$
satisfies (\ref{22}). Indeed, taking into account that,
for every $a\in\Ao$, the equality
$b^*\delta(a)b=\delta(b^*ab)-\delta(b^*)ab-b^*a\delta(b)$ holds,
we have
$$
| \omega_b(\delta(a))|=\left| \omega(b^*\delta(a)b)\right|\leq
\left| \omega(\delta(b^*ab))\right|+\left| \omega(\delta(b^*)ab)\right|
+\left| \omega(b^*a\delta(b))\right|.
$$
Using (\ref{22}) for the first  and introducing $\pi_\omega$ for the
second and the third contributions above, we find that, for every $a \in \Ao$,
\begin{eqnarray*}
| \omega_b(\delta(a))|&\leq&
C\left( \omega(b^*a^*bb^*ab)^{1/2}+ \omega(b^*abb^*a^*b)^{1/2}\right)\\&+&
\left|\ip{\lambda_\omega(ab)}{\lambda_\omega(\delta(b))}\right|+
\left|\ip{\lambda_\omega(\delta(b))}{\lambda_\omega(a^*b)}\right|
\\
& =&C\left(\|\pi_\omega(b)^*\lambda_\omega(ab)\|
+\|\pi_\omega(b)^*\lambda_\omega(a^*b)\|\right)\\ &+&
\left|\ip{\lambda_\omega(ab)}{\lambda_\omega(\delta(b))}\right|+
\left|\ip{\lambda_\omega(\delta(b))}{\lambda_\omega(a^*b)}\right|\\
&\leq &
\left(C\|\overline{\pi_\omega(b)}\|+\|\lambda_\omega(\delta(b)\|\right)
\left(\omega_b(a^*a)^{1/2}+\omega_b(aa^*)^{1/2}\right),
\end{eqnarray*}

\end{proof}
} \vspace{2mm}

The conclusion is therefore that, under mild conditions on $\pi_\omega$, and therefore on $\omega$, both $\delta_{\pi_{\omega}}$ and $\delta_{\pi_{\omega_b}}$
turn out to be spatial so that two different effective hamiltonians $H_\omega$ and $H_{\omega_b}$ do exist, and they are related as in Section I. Once again, the physical contents of the two representations is essentially the same.

\vspace{3mm}

We end this section with some further results on the GNS representations of a quasi *-algebra $(\A,\Ao)$.

Let $(\A,\Ao)$ be a locally convex quasi *-algebra, $\omega_0$ a positive linear functional on $\Ao$ satisfying
(\ref{21}) and $\omega=\widetilde{\omega_0}$ its continuous extension on $\A$. As we have shown, both
$\omega$ and $\omega_b$, $b\in\Ao$, satisfy conditions (L1), (L2) and (L3), and so the GNS-constructions
$(\pi_\omega,\lambda_\omega,\Hil_\omega)$ and $(\pi_{\omega_b},\lambda_{\omega_b},\Hil_{\omega_b})$ are defined.
Let $\tilde\pi_\omega$ and $\tilde\pi_{\omega_b}$ be the closures of $\pi_\omega$ and $\pi_{\omega_b}$,
respectively. In this section we find conditions which imply that $\tilde\pi_\omega$ is unitarily equivalent to
the direct sum of a family of $\tilde\pi_{\omega_b}$, $b\in\Ao$.

{}
\begin{lemma}\label{lemma10} Let $x\in\A$ and $\{x_\alpha\}\subset\Ao$ such
that $\tau-\lim_\alpha x_\alpha =x$, then $\lambda_{\omega}(x_\alpha)=\lambda_{\omega_0}(x_\alpha)\rightarrow
\lambda_\omega(x)$.
\end{lemma}
{}
\begin{proof}
We begin with proving that $\{\lambda_\omega(x_\alpha)\}$ is a Cauchy net in the Hilbert space $\Hil_\omega$:
$$
\|\lambda_\omega(x_\alpha)-\lambda_\omega(x_\beta)\|^2= \omega((x_\alpha-x_\beta)^*(x_\alpha-x_\beta))\leq
p(x_\alpha-x_\beta)^2\rightarrow 0.
$$
Therefore there exists a vector $\xi\in\Hil_\omega$ such that $\lambda_\omega(x_\alpha)\rightarrow \xi$. We now
prove that $\xi=\lambda_\omega(x)$. Indeed we have, for every $c\in\Ao$,
$\ip{\lambda_\omega(x_\alpha)}{\lambda_\omega(c)}\rightarrow \ip{\xi}{\lambda_\omega(c)}$ and, on the other hand,
$\ip{\lambda_\omega(x_\alpha)}{\lambda_\omega(c)}=\omega(c^*x_\alpha)\rightarrow
\tilde\omega(c^*x)=\ip{\lambda_\omega(x)}{\lambda_\omega(c)}$, due to the definition of $\tilde\omega$. Therefore
$\xi=\lambda_\omega(x)$.
\end{proof}


We recall that the weak commutant ${\mathfrak M}'_w$ of a $*-$invariant subset ${\mathfrak M}$ of $\Lc^\dagger(\D,\Hil)$ is defined as
$${\mathfrak M}'_w=\{C\in \B(\Hil): \ip{X\xi}{C^*\eta}=\ip{C\xi}{X^\dagger\eta}, \; \forall X\in {\mathfrak M},\, \xi, \eta\in \D \}.$$
Then we can prove the following

\begin{prop} $\pi_\omega(\A)'_w=\pi_\omega(\Ao)'_w$. \end{prop}

\begin{proof}
The inclusion $\pi_\omega(\A)'_w\subset\pi_\omega(\Ao)'_w$ is clear. To prove the converse inclusion we take
$C\in\pi_\omega(\Ao)'_w$ and $x\in\A$, $c_1,c_2\in\Ao$. Then we have, using the previous Lemma,
\begin{eqnarray*}
\ip{C\pi_\omega(x)\lambda_\omega(c_1)}{\lambda_\omega(c_2)}&=&\lim_\alpha
\ip{C\pi_\omega(x_\alpha)\lambda_\omega(c_1)}{\lambda_\omega(c_2)}\\
&=&\lim_\alpha \ip{C\lambda_\omega(c_1)}{\pi_\omega(x_\alpha^*)\lambda_\omega(c_2)}=
\ip{C\lambda_\omega(c_1)}{\pi_\omega(x^*)\lambda_\omega(c_2)}.
\end{eqnarray*}
\end{proof}

Let $b\in\Ao$. We denote by $P_\omega^b$ the projection of $\Hil_\omega$ onto
$\Hil_\omega^b=\overline{\pi_\omega(\Ao)\lambda_\omega(b)}$. By  Lemma \ref{lemma10} we deduce the following

\begin{lemma}
Suppose that $\pi_\omega(a)$ is bounded for every $a\in\Ao$. Then $\pi_\omega(\A)_w'$ is a von Neumann algebra and
$P_\omega^b\in \pi_\omega(\A)_w'$.
\end{lemma}

Even if ${\pi_\omega}{\upharpoonright_{\Ao}}$ is bounded, $P_\omega^bD(\tilde\pi_\omega)\neq D(\tilde\pi_\omega)$
in general. Hence we introduce the following notion:

\begin{defn} Let $b \in \Ao$. We say that $b$ is a self-adjoint element
for $\pi_\omega$ if $\tilde\pi_{\omega_b}$ is a self-adjoint *-representation of $\A$.
\end{defn}

By (\cite{ait_book}, Theorem 7.4.4) we have the following
\begin{lemma} Let $b$ be a self-adjoint element for $\pi_\omega$. Then

(1) $P_\omega^b\in \pi_\omega(\Ao)'_w
=\pi_\omega(\A)'_w$

(2) $P_\omega^b \D(\tilde \pi_\omega) = \D(\tilde\pi_\omega^b)$.

(3)
$\tilde\pi_\omega^b=(\tilde\pi_\omega)_{P_\omega^b}:=P_\omega^b\tilde\pi_\omega (\cdot) P_\omega^b$.

\end{lemma}
{}
By Proposition 1, $\tilde\pi_\omega^b$ is unitarily equivalent to $\tilde\pi_{\omega_b}$, and by the above Lemma
we have the following result, which answer our original question

\begin{prop}
Suppose that $\tilde\pi_\omega$ is self-adjoint. If there exist a family $\{b_\gamma\}_{\gamma \in \Gamma}$ of
self-adjoint elements for $\pi_\omega$, such that $\{P_\omega^{b_\gamma}\}$ consists of mutually orthogonal
projections and $\sum_{\gamma \in \Gamma}P_\omega^{b_\gamma}=I$, then $ \tilde\pi_\omega$ is unitarily equivalent
to $\displaystyle \bigoplus_{\gamma \in \Gamma}\,\tilde\pi_{\omega_{b_\gamma}}.$

\end{prop}

\section{Local  modifications of states}

We consider now the particular case in which the C*-algebra $\A$ is endowed with a {\em local structure}. Following \cite{brarob} we construct the local C*-algebra as
follows.

Let $\F$ be a set of indexes directed upward and with an orthonormality relation $\perp$ such that (i.) $\forall
\alpha\in\F$ there exists $\beta\in\F$ such that $\alpha\perp\beta$; (ii.) if $\alpha\leq \beta$ and $\beta\perp
\gamma$, $\alpha, \beta, \gamma\in\F$, then $\alpha\perp\gamma$; (iii.) if, for $\alpha, \beta, \gamma\in\F$,
$\alpha\perp\beta$ and $\alpha\perp\gamma$, there exists $\delta\in\F$ such that $\alpha\perp\delta$ and
$\delta\geq \beta, \gamma$.

Let now $\{\A_\alpha(\|.\|_\alpha), \,\alpha\in\F\}$ be a family of C*-algebras with C*-norm $\|.\|_\alpha$, indexed by $\F$, such that (a.) if $\alpha\geq
\beta$ then $\A_\alpha\supset\A_\beta$; (b.) there exists a unique identity $e$ for all $\A_\alpha$'s; (c.) if
$\alpha\perp\beta$ then $xy=yx$ for all $x\in\A_\alpha$, $y\in\A_\beta$. Let further $\A_0:=\cup_\alpha
\A_\alpha$. The uniform completion of $\Ao$ is, as it is well known, the quasi-local C*-algebra\footnote{this terminology is due to the fact that, in concrete applications, $\alpha$ is quite often a
given bounded open region in a $d-$dimensional space} with the norm
$\|\cdot\|$ inherited by the $\|.\|_\alpha$'s. If we take instead the completion of $\Ao$ w.r.t. a locally convex topology $\tau$ which makes the
involution and the multiplications continuous we get, in general, a locally convex quasi *-algebra $\A$ which we
call a {\em quasi-local quasi *-algebra}.

Given $x\in\A_0$, there will be some $\beta\in\F$ such that $x\in\A_\beta$. But of course, $x$ also belongs
to many other $\A_{\beta'}$, for instance to all those algebras which contains $\A_\beta$ as a sub-algebra. For this reason we
introduce a set $J_x$, related to $x\in\A_0$, which is defined as follows: $J_x=\{\alpha\in\F \mbox{ such that }
x\in\A_\alpha \}$. If we now define $\A_\infty=\cap_{\alpha\in\F}\A_\alpha$, then we will work here under the
assumption, which is verified for very general discrete and continuous models \cite{sew}, that $\forall\,
x\in\A_0$, $x\notin \A_\infty$, there exists $\alpha_x\in\F$ such that $\cap_{\beta\in
J_x}\A_\beta=\A_{\alpha_x}$.  We call $\alpha_x$ the {\em  support of $x$}.

{}

The following definition selects states on $\A$ with a {\em reasonable asymptotic behavior}. These states, indeed, factorize on regions far enough from the support of a given element.

\begin{defn}
A state $\omega$ over $\A$ is said to be almost clustering (AC) if, $\forall \,b\in\Ao$ and
$\forall\,\epsilon>0$, there exists $\alpha\in\F$, $\alpha\geq \alpha_b$, such that, $\forall\gamma\perp\alpha$
we have $\left|\omega(ab)-\omega(a)\omega(b)\right|\leq \epsilon \|a\|$, $\forall\,a\in\A_\gamma$.
\end{defn}

Similar definitions are given in many textbooks, like \cite{sew}, \cite{brarob} and \cite{emch}, where the
physical motivations are discussed in  detail. Related to the notion of factorization is also that of {\em local modification} of a given state. Of course, several definitions of local modifications can be introduced. The most natural one is perhaps the following: $\omega'$ is a local modification of $\omega$ if there exists
$\alpha\in\F$ such that $\forall\,\gamma\in\F$, $\gamma\perp\alpha$, $\omega'(a)=\omega(a)$ for all
$a\in\A_\gamma$. This simply implies that, outside a fixed {\em region} $\alpha$, the two states coincide.  However this condition is rather strong and has no counterpart in the existing literature on
this subject and for this reason  will not be considered here. To stay in touch with the existing literature, we rather consider the following definitions.

\begin{defn}
Given two states $\omega$ and $\omega'$ over $\A$, $\omega'$ is said to be a local modification of type 1 (1LM)
of $\omega$ if, calling $\pi_{\omega'}$ and $\pi_\omega$ their associated GNS-representations, $\pi_{\omega'}$ is
unitarily equivalent to a sub *-representation of $\pi_\omega$.

Also, $\omega'$ is said to be a local modification of type 2 (2LM) of $\omega$ if $\forall\,\epsilon>0$, there
exists $\alpha_\epsilon\in\F$ such that, $\forall\gamma\in\F$, $\gamma\perp\alpha_\epsilon$,
$\left|\omega'(x)-\omega(x)\right|\leq \epsilon \|x\|$, $\forall\,x\in\A_\gamma$.
\end{defn}

These definitions are physically motivated
essentially from what is discussed in \cite{sew}. Just to clarify the situation if, for instance, $\omega'$ is a 2LM of $\omega$ then they coincide, but for an error of order $\epsilon$, outside a region whose size is, in general,  proportional to $1/\epsilon$.

There is an apparent difference between the conditions 1LM and 2LM: if $\omega'$ is
a 2LM of $\omega$, then $\omega$ is a 2LM of $\omega'$. This symmetry is not shared by 1LM. We argue that 2LM could be used for the mathematical description of reversible local operations on a given state while 1LM seems to be more appropriate for describing the action of irreversible
operations (like a quantum mechanical measurement).

One immediate consequence of the results of Section II and of these definitions is that if $b\in\A_\alpha$ for
some $\alpha\in\F$ then the state $\omega_b(.)$ is a 1LM of $\omega$. Less trivial is the proof of the following
statement: {\em let the state $\omega$ be AC and $b\in\Ao$ with $\omega(b^\dagger b)=1$. Then $\omega_b$ is a 2LM
of $\omega$}. This is not the end of the story. Indeed, let us suppose that $\omega$ is AC and that $\omega'$ is
a 1LM of $\omega$. Therefore there exists a sequence $\{b_n\}$ of elements of $\Ao$ such that
$\omega'(a)=\lim_{n\to\infty}\omega(b_n^* a b_n)$,  $\forall\,a\in\A$, and the sequence
$\{\pi_\omega(b_n)\xi_\omega\}$ converges in $\Hil_\omega$. We suppose now that there exists $n_0\in\mb{N}$ and
$\lambda\in\F$ such that, for all $n\geq n_0$, $b_n\in\A_\lambda$. Then $\omega'$ is also a 2LM of $\omega$. The
proof of these statements are easy and will be omitted here.

 We end this section, and the paper, with the following example of what a concrete local modification of a state could be.

{\bf Discrete system:} Let $V$ be a finite region of a $d$-dimensional lattice $\Lambda$ and $| V |$ the number
of points in $V$. The local $C^*$-algebra ${\A}_V$ is generated by the Pauli operators $\vec\sigma_p = (\sigma^1_p,
\sigma^2_p, \sigma^3_p)$ and by the unit $2\times 2$ matrix $e_p$ at every point
 $p \in V$. The $\vec\sigma_p$'s are copies of the Pauli matrices localized in $p$.

If $V \subset V^{'}$ and $A_V \in {\A}_V$, then $A_V \rightarrow A_{V^{'}} = A_V \otimes ({{\atop \bigotimes }
\atop{p \in V^{'} \setminus V}} e_p)$ defines the natural imbedding of ${\A}_V$ into ${\A}_{V^{'}}$.

Let ${\vec n}=(n_1, n_2, n_3)$ be a unit vector in ${{\mb R}}^3$, and put $ (\vec\sigma\cdot {\vec n}) = n_1 \sigma^1
 + n_2 \sigma^2 + n_3 \sigma^3.
$ Then, denoting as $Sp(\vec\sigma \cdot {\vec n})$ the spectrum of $\vec\sigma \cdot {\vec n}$, we have $ Sp(\vec\sigma \cdot {\vec n}) = \{ 1,
-1\}. $ Let $|\vec n \rangle\in {\mb C}^2$ be a unit eigenvector associated with $1$.

Let now denote by ${\mathbf n}:= \{ \vec n_p \}_{p\in\Lambda}$  an
infinite sequence of unit vectors in ${{\mb R}}^3$  and $| \mathbf{n} \rangle = {{\atop \bigotimes }\atop{p}} |
\vec n_p\rangle$
 the corresponding unit vector in
the infinite tensor product ${\cal H}_\infty = {{\atop \bigotimes }\atop{p}}
 {\mb C}_p^2$.
We put $ {\A}_0 = \bigcup_V {\A}_V $ and $ {\cal D}^0_{\mbox{\small$\bfn$}} = {\A}_0 |\bfn\rangle $ and we denote the closure
of ${\cal D}^0_{\mbox{\small$\bfn$}}$ in ${\cal H}_\infty$ by
 $ {\cal H}_{\mbox{\small$\bfn$}}$.
As we saw above, to any sequence $ \bfn$  of three-vectors there corresponds a state $|\bfn\rangle$ of the
system. Such a state defines a realization $\pi_{\mbox{\small$\bfn$}}$ of ${\A}_0$ in the Hilbert space ${\cal H}_{\mbox{\small$\bfn$}}$. This
representation is faithful, since the norm completion ${\A}_S$ of ${\A}_0$ is a  simple C*-algebra. A special basis for $ {\cal H}_{\mbox{\small$\bfn$}}$ is obtained from
the {\em ground} state $|\bfn\rangle$ by {\em flipping} a finite number of spins using the following strategy:\\
Let $\vec n$ be a unit vector in ${{\mb R}}^3$, as above, and $ |\vec n\rangle$ the corresponding vector of ${\mb C}^2$. Let us choose two other unit vectors ${\vec n}^1, {\vec n}^2$ so that $({\vec n}, {\vec n}^1, {\vec n}^2)$ form an orthonormal basis of ${{\mb R}}^3$. We put $ {\vec n}_{\pm} = \frac{1}{2} ({\vec n}^1 \pm i{\vec n}^2) $ and define
$ |m,\vec n\rangle := (\vec\sigma \cdot {\vec n}_{-})^m |\vec n\rangle \ \ (m=0,1). $ Then we have
$$
(\vec\sigma \cdot {\vec n}) |m, \vec n\rangle = (-1)^m |m,\vec n\rangle \ \ (m=0,1).
$$
Thus the set $ \left\{|\mathbf{m}, \mathbf{n}\rangle = {{\atop \bigotimes }\atop{p}} |m_p, \vec{n}_p\rangle ;\ m_p = 0, 1,\ \
{\displaystyle \sum_p} m_p < \infty \right\}$ forms an orthonormal basis in ${\cal H}_{\mbox{\small$\mathbf{n}$}}$, \cite{BCS1}.

 The representation $\pi_{\mbox{\small$\mathbf{n}$}}$ is defined on the basis vectors $\{ |\mathbf{m}, \mathbf{n}\rangle \}$ by
 $$
\pi_{\mbox{\small$\mathbf{n}$}} (\sigma^i_p)|\mathbf{m}, \mathbf{n}\rangle= \sigma^i_p \mid m_p, \vec n_p\rangle
  \otimes ({{\atop
\prod }\atop{\scriptstyle{p^{'} \neq p}}} \otimes \mid m_{p ^{'}}, \vec n_{p^{'}}\rangle)\ \ \ (i= 1, 2, 3).
$$
This definition is then extended in obvious way to the whole space ${\cal H}_{\mbox{\small$\mathbf{n}$}}$. It
turns out that $\pi_{\mbox{\small$\mathbf{n}$}}$ is a {\em bounded} representation of $\A_0$ into ${\cal
H}_{\mbox{\small$\mathbf{n}$}}$. More details on this construction, particularly in connection with quasi
*-algebras, can be found in \cite{ctrev,bagtra3}.

Let now $\varphi=\otimes_{j\in\Lambda}\varphi_j$ be a fixed normalized vector in ${\cal
H}_{\mbox{\small$\mathbf{n}$}}$ and $\omega$ the related vector state: if $a\in\A_0$ then
$\omega(a)=<\varphi,\pi_{\mbox{\small$\mathbf{n}$}}(a)\varphi>$. Let now $x=\prod_{p\in\lambda}\otimes x_p$, for
some bounded  subset $\lambda$ in $\Lambda$. Here $x_p$ acts on ${\mb C}_p^2$ and $\lambda$ is the support of
$x$. Let furthermore $\gamma$ be another bounded  subset of $\Lambda$, orthogonal to $\lambda$: this means that
the sets $\lambda$ and $\gamma$ have empty intersection. Then, we fix $b=\prod_{p\in\gamma}\otimes b_p$, where as
before $b_p$ acts on ${\mb C}_p^2$. We further assume that
$<\pi_{\mbox{\small$\mathbf{n}$}}(b)\varphi,\pi_{\mbox{\small$\mathbf{n}$}}(b)\varphi>=1$. Then we can check that
$\omega(a)$ coincides with
$\omega_b(a)=<\pi_{\mbox{\small$\mathbf{n}$}}(b)\varphi,\pi_{\mbox{\small$\mathbf{n}$}}(a)\pi_{\mbox{\small$\mathbf{n}$}}(b)\varphi>$,
and this is true for all possible choices of $a$ and $b$ which are supported in separated regions. So $\omega_b$
is a local modification of $\omega$ in the strongest sense and, in particular, is a 2LM of $\omega$.

\vspace{3mm}

This example shows that the definitions of local modification given here are really physically motivated. States sharing the same properties in the case of continuous physical systems, \cite{sew}, could also be constructed with no major difficulty. To \cite{sew} we also refer for a more physically-minded discussion on 1LM of states.

\section*{Acknowledgements}

This work was partially supported by the Japan Private School Promotion Foundation and partially by CORI, Universit\`a di Palermo. F.B. and C.T. acknowledge
the warm hospitality of the Department of Applied Mathematics of the Fukuoka University. A.I. acknowledges the hospitality of the Dipartimento di Matematica
ed Applicazioni, Universit\`a di Palermo, where this work was completed. \\
{The authors wish to thank the referee for his suggestions and corrections}.

\end{document}